\documentclass{article}
\usepackage{graphicx} 
\usepackage{amsfonts, amsthm, amsmath, amssymb}
\usepackage{natbib}
\usepackage[hidelinks]{hyperref}
\usepackage{xcolor}
\usepackage{nicefrac}
\usepackage{xspace}
\usepackage{makecell}
\usepackage{csquotes}

\theoremstyle{definition}
\newtheorem{definition}{Definition}
\newtheorem{theorem}{Theorem}

\newtheorem{corollary}{Corollary}
\newtheorem{proposition}{Proposition}

\DeclareMathOperator*{\argmax}{arg\,max}
\newcommand{\dt}{\ensuremath{\delta t}\xspace}

\title{Capital Games and Growth Equilibria} 
\author{Ben Abramowitz}
\date{July 2025}

\begin{document}

\maketitle

\begin{abstract}
    We introduce capital games, which generalize the definition of standard games to incorporate dynamics. In capital games, payoffs are in units of capital which are not assumed to be units of utility. The dynamics allow us to infer player utilities from their individual payoffs and linearizable capital dynamics under the assumption that players aim to maximize the time-average growth rate of their capital.
\end{abstract}

\section{Introduction}
Based on the seminal work of John von Neumann and Okar Morgenstern, Game Theory assumes that player behavior is based on a desired to maximize \textit{utility}. The challenge in practice is to determine what player utilities are, because outcome payoffs may not be in units of utility. Here we show how tools from Ergodicity Economics allow us to generalize the original definition of games by incorporating dynamics, derive utility functions from player payoffs, and compute equilibria in a generalized model of games.

First we will introduce the background concepts of utility functions, the von Neumann-Morgenstern Utility Theorem, standard games, gambles, and gamble problems, using the infamous coin flip example of~\cite{peters2016evaluating} as a reference point. Then we will provide a generalized definition of games that we call capital games that incorporate dynamics, and a generalized definition of equilibria called growth equilibria. Finally, we show how computing the growth equilibria of a capital game with linearizable dynamics is equivalent to computing Nash equilibria in a corresponding standard game.

\section{Preliminaries}

\subsection{Utility Functions}
In their seminal book \textit{The Theory of Games and Economic Behavior} John von Neumann and Oskar Morgenstern (vNM) proved the von Neumann–Morgenstern Utility Theorem.
The vNM Utility Theorem is a statement about preferences over lotteries, where a lottery is a discrete probability distribution over a given set of outcomes.
The theorem tells us that for any ordinal preference over lotteries that satisfies basic axioms, there is a way to assign real values, or \textit{utilities}, to the outcomes of the lotteries such that the ordering of expectation values of utilities matches the preference relation.

\paragraph{vNM Utility Theorem}
Define a lottery to be a discrete probability distribution over a set.
We will denote that lottery $A$ is strictly preferred to lottery $B$ by $A \succ B$, indifference between the lotteries as $A \sim B$, and weak preference by $A \succeq B$. We now define four axioms that characterize preference relations.

\begin{itemize}
    \item Completeness: For every pair of lotteries, either $A \succeq B$ or $B \succeq A$
    \item Transitivity: If $A \succeq B$ and $B \succeq C$ then $A \succeq C$
    \item Continuity: If $A \succeq B \succeq C$ then there exists a probability $p \in [0,1]$ such that $pA + (1-p)C \sim B$ 
    \item Independence: For all lotteries $C$ and probabilities $p \in [0,1]$, $A \succeq B$ if and only if $(1-p)A + pC \succeq (1-p)B + pC$
\end{itemize}

Without loss of generality, any time we talk about two or more lotteries we will assume they are defined over the same set of outcomes, but may assign a probability of 0 to certain outcomes.

\begin{theorem}[vNM Utility Theorem]
    For any preference relation satisfying Completeness, Transitivity, Continuity, and Independence, there is a utility function $u$ that assigns a real number to every possible outcome of the lotteries such that for any two lotteries $A,B$, $A \succ B$ if and only if $E[u(A)] > E[u(B)]$.
\end{theorem}

The work of vNM left open the question of what our preference relation over lotteries should be in any given setting. That is the question to which this paper is devoted.

\subsection{The PGM Coin Flip}
Before going further, it is useful to review the infamous coin flip example from \cite{peters2016evaluating}(PGM).

Suppose you have the opportunity play a simple single-player game in which you must pay some amount of money $x$ to play, and then get to flip a fair coin. If the coin lands heads, then you will gain 50\% of your wager for a payoff of $1.5x$, but if the coin lands tails you will lose 40\% of your wager for a payoff of $0.6x$.
You do not have to play the game. Your other option is to simply keep your $x$ units of money.
Naturally, we assume you would strictly prefer winning and increasing your wealth over not playing, and you would strictly prefer not playing to losing the bet and decreasing your wealth.

Should you choose to play or to keep your money? Deciding whether to play means choosing between two lotteries with the outcome set \{Heads, Tails, No Flip\} where the outcomes come with monetary values $(1.5x, 0.6x, x)$, respectively. One lottery assigns the probabilities $(0.5, 0.5, 0)$ to the outcomes and the other lottery assigns the probabilities $(0,0,1)$.

For the moment, let's try treating the monetary payoffs as utilities. Then we can take the expectation value for each lottery and compare them. Playing the game has expectation value $\frac{1.5x + 0.6x}{2} = 1.05x$, and keeping your money has the trivial expectation value of $x$ since it is deterministic. Since $1.05x > x$, we are tempted to conclude that playing the game is better than keeping our money because it appears to offer greater utility. If this is true, then arguably we should play the game as many times as we possibly can.

Imagine that you were to play the PGM coin flip game many, many times in \textit{parallel}. After all of the games conclude, you sum your total winnings (or loses) to get your final payoff. Let $2n$ be the number of games you play. In the limit as $n \rightarrow \infty$, you expect the coin to land heads in approximately half of the games and tails in approximately half of the games. Your expected total payoff therefore converges to $x\frac{1.5n + 0.6n}{2n} = 1.05x$. You expect to make money, and therefore playing the games many, many times in parallel is better than holding onto your money, as long as you have at least $2xn$ available to play all of the games.

Now imagine instead that you have the opportunity to play the game not just once, but many, many times in \textit{sequence}. You start with an initial wager of $x$. In each round the coin is flipped your wealth grows or shrinks and becomes your wager in the next round, with no opportunity to add or remove money from the game. In the first two rounds if you flip heads twice in a row then your wealth would become $x(1.5)^2 = 2.25x$. If the coin flips tails twice in a row, then your wealth would be $x(0.6)^2 = 0.36x$ after two rounds. And finally, if you flip heads once and tails once, in either order, then your wealth would be $x(1.5)(0.6) = 0.9x$ after two rounds.

After an even number of rounds $2n$ of this sequential game, if you have flipped an equal number of heads and tails, then your wealth has decreased from $x$ to $x(0.9)^n$. You would need to consistently flip more heads than tails just to break even. By the law of large numbers, as $n$ increase, you expect roughly the same number of heads and tails, and therefore you should expect your wealth to decrease at an exponential rate. Playing the sequential game is a bad bet and it would be better not to play and just keep your money.

We observe from the PGM coin flip example that if you get to choose between the two lotteries many times in parallel, then the payoffs act like utilities, but when the games are played in sequence, the payoffs are no longer utilities.
When the games are sequential, your wealth exhibits multiplicative dynamics, as opposed to the additive dynamics from the parallel games, and this leads to different utility functions.

In the case of a one-shot game -- only a single coin flip -- should we choose to play? The answer is indeterminate because we have not specified the broader dynamics within which our decision is made. The answer, and our utility function, depend on what we expect to happen to our wealth in the future, including the future decisions we expect to face. These dynamics can be modeled by stochastic processes.

\subsection{Non-Utility Observables}
Lotteries are defined as discrete probability distributions over sets of outcomes, but those outcomes do not necessarily have real values associate with them. The outcomes can be anything.
An \textit{observable} is a mapping of lottery outcomes to real values. Utility functions are observables but, crucially, not all observables are utility functions. In other words, vNM utility is defined to be the variable whose expectation value we want to maximize, but that is not the case for all observables.

\begin{quote}
    \textit{We have practically defined numerical utility as being that thing for which the calculus of mathematical expectations is legitimate.} -\cite{von2007theory}.
\end{quote}


Sometimes the outcomes of lotteries have real numbers naturally associated with them, e.g., monetary amounts. The monetary values of the outcomes are an observable, but that does not mean they are utilities.
In the PGM coin flip example, the monetary payoffs are an observable, but whether that observable is a utility function depends on the broader dynamics.
%

When outcomes are not associated with an observable, we have no method for inferring preference relations over lotteries, and hence no general method of constructing utility functions. Our goal is to find a fairly general method for inferring utility functions from observables.

\paragraph{Observables as Preferences Over Outcomes}
For our purposes, we will assume that the observable for any set of outcomes corresponds to a preference over outcomes. A larger observable value means the outcome is preferred, and outcomes with the same observable value are equally preferable. For example, bigger monetary payoffs are preferred to smaller monetary payoffs. In other words, we will interpret our observables as implying a preference ordering over outcomes, while utility functions are a special subclass of observable that also imply a preference ordering over lotteries. We will frequently refer to observables as payoffs.







\subsection{Games and Nash Equilibrium}
We now provide a background on key concepts, definitions, and results from Game Theory.

\begin{definition}[Standard Game]
A standard game is a finite, n-person game characterized (in normal form) by a tuple $(N, A, u)$ where:
\begin{itemize}
    \item N is a finite set of $n \geq 1$ players, indexed by $i$;
    \item $A = (A_1, \ldots, A_n)$ where $A_i$ is a finite set of actions (or pure strategies) available to player $i$
    \item $u = (u_1, \ldots, u_n)$ where $u_i : A_1 \times \ldots \times A_n \rightarrow \mathbb{R}$ is a real-valued function that determines the utility of every player based on the action profile
\end{itemize}
\end{definition}

\paragraph{Action Profiles and Payoffs}
We let $\mathcal{A} = A_1 \times \ldots \times A_n$ denote the set of all possible action profiles, where each action profile $a = (a_1, \ldots, a_n) \in \mathcal{A}$ specifies a single action for every player. We can see by the definition of $u_i$ that the utility of each player depends on the collective actions of all players, not just their own (when $n > 1$).

\paragraph{Mixed Strategy Profiles}
Players may use mixed strategies, which are probability distributions over their available actions. For any set $X$, let $\Delta(X)$ denote the set of all probability distributions over $X$. Then the set of strategies $S_i$ available to player $i$ is $S_i = \Delta(A_i)$. Players select their strategies simultaneously and independently. 

A mixed strategy profile is given by $s = (s_1, \ldots, s_n) \in \mathcal{S} = S_1 \times \ldots \times S_n$. We denote by $s_i(a^j)$ the probability that strategy $s_i$ assigns to taking action $a^j \in A_i$.
Given a mixed strategy profile $s$, we define $s(a) = \prod_{i \in N} s_i(a_i)$ to be the probability of the action profile $a = (a_1, \ldots, a_n)$ being realized under $s$.
The support of $s_i$ is all actions $a^j \in A_i$ such that $s_i(a^j) > 0$, and similarly, the support of $s$ is all action profiles $a$ such that $s(a) > 0$.

All players are assumed to have complete and perfect information, so they know the strategy spaces and utility functions of the other players, and it is common knowledge that all players have complete and perfect information.

\begin{definition}[Expected Utility]
    Given normal-form game $(N, A, u)$, and a mixed strategy profile $s = (s_1, \ldots, s_n)$, the expected utility of player $i$ is defined as:
    $$E[u_i|s] = \sum\limits_{a \in \mathcal{A}} u_i(a) s(a)$$
\end{definition}

The \textit{best response} of a player is to maximize their expected utility.
We define $s_{-i} = s \backslash s_i$ to be the strategy profile excluding player $i$, and can therefore use the shorthand $s = (s_i, s_{-i})$. Similarly, we let $a_{-i}$ be the set of all actions taken by players other than player $i$, so that action profile $a$ can be represented by $(a_i, a_{-i})$.

\begin{definition}[Response]
    A response is a choice of strategy $s_{i}$ given a profile $s_{-i}$ of the strategies of all other players.
\end{definition}

\begin{definition}[Best Response]
    Player $i$'s best response to strategy profile $s_{-i}$ is a mixed strategy $\bar{s}_i \in S_i$ such that $E[u_i|(\bar{s}_i,s_{-i})] \geq E[u_i|(s_i, s_{-i})]$ for all $s_i \in S_i$.
\end{definition}

A Nash equilibrium is a strategy profile such that no player $i$ can increase their expected utility by unilaterally changing their strategy. In other words, $s$ is a Nash equilibrium if for all players $i$, $s_i$ is a best response to $s_{-i}$.

\begin{definition}[Nash Equilibrium]
    A Nash equilibrium is a strategy profile $s$ such that $s_i$ is a best response to $s_{-i}$ for all players $i \in N$, maximizing their expected utility.
\end{definition}

\begin{theorem}[Existence of Nash Equilibria]
    All standard games have at least one Nash equilibrium~\citep{nash1950equilibrium}.
\end{theorem}

Although Nash equilibria always exist, computing an equilibrium of a standard game is known to be PPAD-complete~\citep{daskalakis2009complexity}.



\subsection{Gambles and Best Responses}
We now introduce key concepts from Ergodicity Economics, which we will apply to the study of games~\citep{eebook}.

\begin{definition}[Gamble]
    A gamble is a tuple $(Q, \dt)$ where $Q$ is a discrete random variable that takes one of $K$ real values $\{q_1, \ldots, q_K\}$ each with probability $\{p_1, \ldots, p_K\}$, and \dt is a unit of time called duration.
\end{definition}

A gamble is a lottery whose outcomes are real-valued (i.e., an observable) and which takes place over some amount of time.

\begin{definition}[Gamble Problem]
    A gamble problem is a problem of choosing between two or more available gambles.
\end{definition}

\begin{proposition}
    In a standard game, the problem of a player choosing their best response $s_i \in S_i$ given $s_{-i}$ and $u_i$ is a gamble problem if we let $\dt$ be the duration of the game.
\end{proposition}

\begin{proof}
    Given $s_{-i}$, each choice of strategy $s_i \in S_i$ creates a different strategy profile $s = (s_i, s_{-i})$, which assigns a probability $s(a)$ to each possible action profile $a$, of which there are finitely many. Since $u_i$ is deterministic, this means that $s$ also induces a probability distribution over the finite number of possible utilities of player $i$. Of course the utilities of $i$ are real-valued, and $\sum_{a \in A} s_i(a) = 1$. Therefore choosing a strategy $s_i$ given $s_{-i}$ means choosing a gamble defined by possible realizations $\{u_i(a)\}_{a \in \mathcal{A}}$ with associated probabilities $\{s(a)\}_{a \in \mathcal{A}}$, and \dt is the duration of the game.
\end{proof}

From the definition of a vNM utility function, we have built in the assumption that the player wishes to maximize the expected value of utility.
In practice, we are given gambles with the values $\{q_1, \ldots, q_K\}$, but we cannot assume these values are the same units as our utility.

As \cite{eebook} argue, and as we saw with the PGM coin flip, the definition of a gamble on its own does not provide sufficient information to choose between gambles. This equates to saying that a set of gambles does not given us enough information to uniquely determine a utility function. 

We cannot generally compare possible gambles without (1) knowing how the outcome affects the future, including future decisions; (2) our circumstances, i.e., in the relevant units ($q_i$), how much are we initially endowed with before the gamble; and (3) our personal preferences, e.g., risk tolerance. 
We will put aside the psychological third aspect for our purposes, and extend our treatment of games based on the first two; by having players start with some initial endowment of an observable of interest (\textit{capital}) and treat each game as a single decision in a (potentially infinite) sequence of decisions each player faces. This sequence of future decisions characterizes the dynamics of the player's ``wealth" in units of the observable.

In practice, the realization of each outcome of a gamble could have a different duration. For brevity we will assume all outcomes of a gamble have the same duration.

\subsection{Contributions}
Our goal is to derive utility functions from the payoffs of a sets of lotteries so that the tools of Game Theory can be brought to bear in reasoning about player behavior and equilibria. The key challenge is that given a set of lotteries with real-valued outcomes, there is no implied preference ordering over the lotteries, and therefore no implied utility function, even if the observable implies a unique preference ordering over the outcomes.

By adding the dimension of time, we cast the problem of choosing between lotteries as a gamble problem. Thus, a player's choice of strategy in a game (with some duration) becomes a gamble problem. We generalize the definition of standard games to incorporate dynamics, and call these \textit{capital games}. Similarly, we generalize the definition of Nash equilibrium to define \textit{growth equilibrium} in capital games. When dynamics are linearizable, we show that the growth equilibria of capital games correspond to the Nash equilibria in a standard game. The conversion from positive capital games to standard games is shown to be reversible, so their equilibria coincide exactly. 

Our results depend on the players' decision criterion for choosing between gambles based on dynamics, which we assume is based on the time-average growth rate of their capital. We aim to introduce and formalize these concepts without getting into the formal details of ergodicity.

\section{Capital Games}
We now define capital games, which are like standard games except that (1) payoffs are in arbitrary units of ``capital" rather than units of ``utility", (2) every player has some initial capital endowment, (3) every player's payoffs are characterized by dynamics which take place over some amount of time.

\begin{definition}[Capital Game]
    A finite capital game is characterized by a tuple $(N, W, A, x, D, f)$ where
    \begin{itemize}
    \item N is a finite set of $n \geq 1$ players, indexed by $i$;
    \item $A = (A_1, \ldots, A_n)$ where $A_i$ is a finite set of actions available to player $i$
    \item $x = (x_1, \ldots, x_n)$ where $x_i : A_1 \times \ldots \times A_n \rightarrow \mathbb{R}$ is a real-valued function that determines the payoff of every player based on the action profile in units of capital
    \item $W = (w_1, \ldots, w_n) \in \mathbb{R}^n$ is an initial endowment for each player, in units of capital.
    \item $D = (\dt_1, \ldots, \dt_n)$ the duration of the game for each player
    \item $f = (f_1, \ldots, f_n)$ where the $f_i$ are each deterministic functions that captures player $i$'s capital dynamics $f_i(x_i(a), w_i, \dt_i) \in \mathbb{R}$.
\end{itemize}
\end{definition}


\begin{definition}[Dynamics Linearization]
    Function $v_i$ is a linearization of capital dynamics $f_i$ if $v_i(f_i(x_i(a), w_i, \dt_i)) = \frac{v_i(x_i(a)) - v_i(w_i)}{\dt_i}$ for all $a \in \mathcal{A}$.
\end{definition}

For brevity we will denote capital games by $(N,W,A,x,f)$ and assume the game duration is equal for all players who act simultaneously and receive their payoffs simultaneously, so the duration is a constant normalized to 1.
We will therefore also use the shorthand $f_i(x_i(a), w_i)$ for capital dynamics.

We will assume that the capital dynamics of all players are linearizable, meaning that a dynamics linearization exists for each of them. Many different linearizable capital dynamics are possible. The two simplest dynamics are additive dynamics $f_i(x_i(a), w_i) = x_i(a) - w_i$ and multiplicative dynamics $f_i(x_i(a), w_i) = \frac{x_i(a)}{w_i}$, as we saw with the PGM coin flip. For additive dynamics $v_i(x)=x$ is a linearization because $f_i$ is already a linear function. For multiplicative dynamics, $v_i(x) = \ln x$ is a linearization because $\ln \frac{x_i(a)}{w_i} = \ln x_i(a) - \ln w_i$, which is defined for all positive capital games.

\begin{definition}[Positive Capital Game]
    A capital game $(N, W, A, x, f)$ is positive if $w_i > 0$ and $x_i(a) > 0$ for players $i$ and action profiles $a$.
\end{definition}

\section{Best Responses and Growth Equilibria}
Every player's \textit{best response} to the strategies of other players is the strategy that maximizes the expected value of their utility. However,  maximizing expected value is only the right decision criteria for choosing among strategies when their payoff is expressed in units of utility.
For general payoffs which may not be utilities, players determine their best responses in accordance with the decision axiom, using all available information.

\begin{proposition}[Decision Axiom \citep{eebook}]
    Players seek to maximize the time-average growth rate $\bar{g}_i$ of their capital.
\end{proposition}

The time-average growth rate of capital is ergodic, and in our setting
$$\bar{g}_i = E[v_i(f_i(x_i(a), w_i, \dt_i))|s]$$

\begin{definition}[Best Response]
Given a capital game $(N,W,A,x,f)$ and opponents' strategy profile $s_{-i}$, the best response of player $i$ is the strategy $s_{i}$ that maximizes the time-average growth rate of player $i$'s capital:
$$\bar{s}_i = \argmax\limits_{s_i \in S_i} E[v_i(f_i(x_i(a), w_i, \dt_i))|s]$$
\end{definition}

The time-average growth rate of player $i$ in a game depends on the full strategy profile $s$. We can define a general growth equilibrium $s^*$ when all players are playing their best responses, maximizing the time-average growth of their capital conditioned on the strategies of the other players.

\begin{definition}[Growth Equilibrium]
    Given an capital game $(N,W,A,x,f)$, a growth equilibrium is a strategy profile $s^* = (s^*_i, s^*_{-i})$ such that $s^*_i$ is a best response to $s^*_{-i}$ for all players $i \in N$.
\end{definition}

Recall that we assume complete information, so each player must know the capital dynamics and duration of the other players to infer their equilibrium strategies.

Capital dynamics can be inhomogenous, or different for every player.
Consider a gambler playing roulette in a casino. The gambler controls what fraction of their endowment they bet in each roll of the roulette wheel. Suppose the gambler bets a fixed fraction of their wealth on every game they play. Then the player's wealth faces multiplicative dynamics, but the casino's wealth experiences a different dynamics because the bet sizes are not a fixed fraction of the casino's wealth.

We prove that for all positive capital games where each player's capital dynamics admit a linearization, there is a correspondence between their growth equilibria and the Nash equilibria of a standard game, assuming players' capital dynamics are common knowledge.

\begin{theorem}[Equilibrium Correspondence]
    Let $G = (N, W, A, x, f, D)$ be a positive capital game, where $\dt_i=1$ and capital dynamics $(f_1,\ldots,f_n)$ can be linearized by $(v_1, \ldots, v_n)$, for all $i \in N, a \in \mathcal{A}$. Let $G' = (N, A, u)$ be the standard game where $u_i(a) = v_i(f_i(x_i(a),w_i))$ for all $i \in N, a \in \mathcal{A}$. Then the Nash equilibria of $G'$ are exactly the growth equilibria of $G$.
\end{theorem}

\begin{proof}
Given $s_{-i}$, the set of best responses of player $i$ is

\begin{align*}
\bar{S}_i 
&= \argmax\limits_{s_i \in S_i} \sum\limits_{a \in \mathcal{A}} v_i(f_i(x_i(a),w_i))s(a)\\
&= \argmax\limits_{s_i \in S_i} \sum\limits_{a \in \mathcal{A}} u_i(a)s(a) = \bar{S}'_i\\
\end{align*}

In a growth equilibrium $s^*$ of $G$, all players are playing their best responses to one another under their respective capital dynamics, which means that in $G'$ all players are playing their best responses to one another in $s^*$, and $s^*$ is therefore a Nash equilibrium.
\end{proof}

\begin{corollary}[Existence of Growth Equilibria]
    Every positive capital game with linearizable dynamics has at least one growth equilibrium.
\end{corollary}

\begin{corollary}[Complexity of Computing Growth Equilibria]
    Computing a growth equilibrium for positive capital games is PPAD-complete.
\end{corollary}



\section{Additive and Multiplicative Dynamics}
For any positive capital game, we can use dynamics linearization to create a standard game whose Nash equilibria are exactly the growth equilibria in our original game. In general, this does not require the players to have the same capital dynamics or the same linearization functions. But here, we will show the construction and correspondence in detail when all players have additive and multiplicative dynamics.

\subsection{Additive Dynamics and Nash Equilibria}
Let's first look at the case where all players have additive capital dynamics. Under additive capital dynamics, $v_i(x)=x$ is a linearization because $f_i$ is already linear. That means capital is utility and utility is capital. Every player's best response is to maximize the expected value of their capital. With additive dynamics, the growth equilibrium is the Nash equilibrium.

%
Given $s_{-i}$, each best response is 

\begin{align*}
\bar{s}_i 
&= \argmax\limits_{s_i \in S_i} \sum\limits_{a \in \mathcal{A}} v_i(f_i(x_i(a),w_i))s(a)\\
&= \argmax\limits_{s_i \in S_i} \sum\limits_{a \in \mathcal{A}} (v_i(x_i(a)) - v_i(w_i))s(a)\\
&= \argmax\limits_{s_i \in S_i} \sum\limits_{a \in \mathcal{A}} (x_i(a)-w_i)s(a)\\
&= \argmax\limits_{s_i \in S_i} \sum\limits_{a \in \mathcal{A}} x_i(a)s(a)\\
\end{align*}

Given a positive capital game $G$, we can construct a standard game $G'=(N,A,u)$ by letting $u_i(a) = x_i(a) - w_i$ for all players and action profiles. Note that the players, action spaces, and strategy spaces are the same as in the capital game. We can see immediately that the best responses in the capital game are also the best responses in the standard game:

\begin{align*}
\bar{s}'_i = \bar{s}_i
&= \argmax\limits_{s_i \in S_i} \sum\limits_{a \in \mathcal{A}} u_i(a)s(a)
\end{align*}

 When all players respond this way, the growth equilibria in the capital game $G$ are the Nash equilibria in the standard game $G'$.

 This argument is reversible. Given a standard game $G'$ we can construct a positive positive capital game $G$ by selecting $w_i >0$ for each player, $x_i(a) = u_i(a)+w_i$, and $f_i(x_i(a), w_i) = x_i(a)-w_i$ for all players and action profiles. Formally, we must further specify $\dt_i = 1$ for all players. We can now follow the steps of the argument above in reverse to see that the growth equilibria of $G$ are the Nash equilibria of $G'$.

 Notice that for any standard game with additive dynamics we can set the endowments of each player to be whatever we want in the capital game. Thus, we can choose them to be the same as any positive capital game of interest. Therefore, there is a one-to-one correspondence between growth equilibria in positive capital games with additive dynamics for all players and Nash equilibria in the corresponding games.

\subsection{Multiplicative Dynamics and Kelly Equilibria}
Let's look at the case where all players have multiplicative capital dynamics with linearization $v(x) = \ln x$. Player $i$ chooses their best response based on maximizing their time-average growth rate under multiplicative dynamics.
%
Given $s_{-i}$, each best response is 

\begin{align*}
\bar{s}_i 
&= \argmax\limits_{s_i \in S_i} \sum\limits_{a \in \mathcal{A}} v_i(f_i(x_i(a),w_i))s(a)\\
&= \argmax\limits_{s_i \in S_i} \sum\limits_{a \in \mathcal{A}} (v_i(x_i(a)) - v_i(w_i))s(a)\\
&= \argmax\limits_{s_i \in S_i} \sum\limits_{a \in \mathcal{A}} (\ln(x_i(a)) - \ln(w_i))s(a)\\
&= \argmax\limits_{s_i \in S_i} \sum\limits_{a \in \mathcal{A}} \ln(x_i(a))s(a)\\
&= \argmax\limits_{s_i \in S_i} \sum\limits_{a \in \mathcal{A}} \ln(x_i(a)^{s(a)})\\
&= \argmax\limits_{s_i \in S_i} \ln \prod\limits_{a \in \mathcal{A}} x_i(a)^{s(a)}\\
&= \argmax\limits_{s_i \in S_i} \prod\limits_{a \in \mathcal{A}} x_i(a)^{s(a)}\\
\end{align*}

which is well-defined because $\ln(x)$ is defined for all $x>0$, and we are considering positive capital games.

Given a positive capital game $G = (N, W, A, x, f)$ where all players have multiplicative capital dynamics (and all $\dt_i = 1$), we can create a standard game $G' = (N,A,u)$ where $u_i(a) = \ln x_i(a) - \ln w_i$ for all players $i$ and action profiles $a$. As with additive dynamics, the players, action spaces, and strategy spaces are the same, and now the growth equilibria of $G$ are all Nash equilibria in $G'$.

Once again we can reveres the process. Given a standard game $G' = (N,A,u)$, create a positive capital game $G = (N, W, A, x, f)$ where $w_i > 0$, $x_i(a) = e^{u_i(a)}w_i$, $f_i(x_i(a),w_i) = \frac{x_i(a)}{w_i}$, and $\dt_i = 1$, for all players and action profiles. The capital game is guaranteed to be positive because $e^x > 0$ for all $x \in \mathbb{R}$. We can follow the steps above in reverse to see that the growth equilibria of $G$ are the Nash equilibria of $G'$.

\section{Pure Growth Equilibria}
 Just as we can define pure Nash equilibria, we can define pure growth equilibria where players are restricted to pure strategies rather than mixed strategies.

\begin{definition}[Growth Rate]
    Given a monotonically increasing function $v$, and action profile $a$, we define the growth rate of the utility of player $i$ to be $g^v_i(a) = \frac{v(x_i(a))-v(w_i)}{\dt_i}$.
\end{definition}

\begin{definition}[Best Response in Pure Strategies]
    When all players are restricted to pure strategies (i.e., actions), a best response to $a_{-i}$ is $$a^*_i = \argmax\limits_{a_i \in A_i} g^v_i(a_i, a_{-i})$$ which does not depend on the choice of function $v$.
\end{definition}

\begin{definition}[Pure Growth Equilibrium]
    A pure growth equilibrium is a strategy profile, which is also an action profile, $a^*$ such that every player $i$ is playing a best response in pure strategies $a_i^*$ to $a^*_{-i}$.
\end{definition}

\begin{theorem}
    If $a^*$ is a pure growth equilibrium for any capital dynamics, then it is also a pure growth equilibrium for all other capital dynamics.
\end{theorem}

\begin{proof}
    If there exists a pure growth equilibrium $a^*$, then each strategy $a_i$ is a choice between actions, rather than probability distributions, and each player $i$ is choosing an action deterministically that maximizes their payoff. This maximizes their deterministic growth rate, and therefore maximizes their time-average growth rate under any assumed dynamics.
\end{proof}

Preferences over degenerate, deterministic lotteries are exactly preferences over outcomes, so the observable (capital) is a vNM utility.

\section{Revisiting the PGM Coin Flip}
In a non-cooperative capital game, a best response is one that maximizes time-average growth rate of capital. To translate from capital to vNM utility we need linearization of dynamics. But where to these dynamics come from?

Let's look back at the infamous coin flip from~\cite{peters2016evaluating}. Consider a player with endowment $w = 100$, action space $A = \{a_1, a_2\}$, and payoffs $x(a^1)=150$ and $x(a^2)=60$. If dynamics are additive, then $f(x(a_1),w) = 150 - 100 = 50$ and $f(x(a_2),w) = 60 - 100 = -40$, and the corresponding standard game has $u_i(a) = f(x(a),w)$, so the utilities are $50$ and $-40$. Thus, if we play the game, we have an expected utility of 5. This corresponds to the expected gain of $\frac{150 + 60}{2} - 100 = 5$, so it is better to play the game than to hold onto cash. On the other hand, if the dynamics are multiplicative, then $f(x(a_1),w) = \frac{150}{100} = 1.5$ and $f(x(a_2),w) = \frac{60}{100} = 0.6$, and the corresponding standard game has $u_i(a) = \ln f(x(a),w)$, so the utilities are $\ln 1.5 \approx 0.405$ and $\ln 0.6 \approx -0.511$. The expected utility is $\frac{0.405 - 0.511}{2} = -0.053$, representing an expected loss of utility. It s better not to play. The key observation here is that the payoffs and endowments together do not imply the dynamics, and specifying the dynamics allows us to derive the vNM utilities for players who wish to maximize the time-average growth rate of their capital.

So why should a player have additive or multiplicative dynamics or a totally different dynamics altogether? In standard games, the meaning of utilities could be ignored because regardless of what it represents the player seeks to maximize its expected value. This is built into the definition of a utility. But in capital games we cannot make this assumption. The dynamics depend on factors not otherwise captured in the definition of a capital game. The answer is application-specific.

\section{Conclusions}
We introduce capital games to extend the definition of standard games to incorporate dynamics, and show that the growth equilibria of capital games with linearizable dynamics correspond to the Nash equilibria of standard games when players seek to maximize the time-average growth rate of their capital.
Broadly, we have shown that linearizable dynamics allow us to infer unique preference orders over lotteries with real-valued payoffs, and therefore derive utility functions.
Notably, the units of utility function may not be the same units as the original capital payoffs, e.g., utilities may be in log-dollars instead of dollars.

Here we have not captured uncertainty about dynamics, or dealt with non-linearizable dynamics, with are important directions for future work. With deterministic dynamics, the only source of uncertainty comes from players' mixed strategies. 

Determining the appropriate dynamics for a given application is an important general problem for Mechanism Design and Decision Theory. 
One theoretical setting in which dynamics should be clear is in infinitely repeated games where a player's capital at the end of one stage game becomes their endowment at the beginning of the next. For example, if the coin flip game is repeated many times in sequence with multiplicative dynamics, then the payoffs keep changing in every round but the ratio between payoff and endowment stays the same for every action profile.

\bibliographystyle{plainnat}
\bibliography{bib}
\end{document}